\documentclass{article}
\usepackage[utf8]{inputenc}
\usepackage[super]{nth}
\usepackage{amsmath}
\usepackage{amssymb}
\usepackage{amsthm}
\usepackage{mathtools}
\usepackage{tikz-cd}
\usepackage{xcolor}
\usepackage{authblk}
\usepackage{url}
\usepackage{hyperref}

\newtheorem{theorem}{Theorem}[section]
\newtheorem{corollary}[theorem]{Corollary}
\newtheorem{prop}[theorem]{Proposition}
\newtheorem{lemma}[theorem]{Lemma}

\newcommand{\id}{\mathrm{id}}

\newcommand{\Hom}{\mathrm{Hom}}
\newcommand{\End}{\mathrm{End}}

\newcommand{\naturalNumbers}{\mathbb{N}}

\newcommand{\complexNumbers}{\mathbb{C}}

\newcommand{\CJ}{Choi-Jamio\l{}kowski }

\allowdisplaybreaks

\title{On the Extremality of the Tensor Product of Quantum Channels}
\author{James Miller Simeão Toledo da Silva \footnote{jmsarxiv@outlook.com}}
\affil{University of São Paulo, IME, Brazil}

\begin{document}
	
\maketitle
\begin{abstract}
	Completely positive and trace preserving (CPT) maps are important for Quantum Information Theory, because they describe a broad class of of transformations of quantum states. There are also two other related classes of maps, the unital completely positive (UCP) maps and the unital completely positive and trace preserving (UCPT) maps. For these three classes, the maps from a finite dimensional Hilbert space $X$ to another one $Y$ is a compact convex set and, as such, it is the convex hull of its extreme points. The extreme points of these convex sets are yet not well understood. In this article we investigate the preservation of extremality under the tensor product. We prove that extremality is preserved for CPT or UCP maps, but for UCPT it is not always preserved.
\end{abstract}

\section{Introduction}

Quantum Information Theory is a theory of storage and processing of information in the form of quantum states. Completely positive and trace preserving maps \cite{nielsen_chuang_2010} are used to describe open quantum systems, and is also a broad class of transformations of quantum states. Being a central object of Quantum Information Theory, understanding its mathematical structure is important for the development of the subject. We'll consider not only CPT maps, but also related classes of maps, the unital completely positive (UCP) maps and the unital completely positive and trace preserving (UCPT) maps.

CPT maps from a finite dimensional Hilbert space $X$ to another $Y$ forms a compact convex set $CPT(X,Y)$. The same is true for $UCP$ or $UCPT$ maps. Every compact convex set is the convex hull of its extreme points. There are characterizations of these set of extreme points due to Choi \cite{CHOI1975285} (see theorem \ref{theorem: characterization of extreme cpt maps}) and Laudau and Streater\cite{LANDAU1993107} (see theorem \ref{theorem: characterization of extreme points for ucpt}), but the extreme points are yet not well understood. An important question is the stability of extremality under some operations. These include the operations of composition and tensor product. It is already known that the composition of two extreme CPT maps need not be extreme \cite{stackExchange}. But, apparently, it wasn't yet answered the question about the extremality of the tensor product of extreme CPT, UCP or UCPT maps (\cite{nlab}). In this article we prove that extremality is preserved by the tensor product in the cases of CPT or UCP maps. For UCPT maps extremality is preserved if one of the Hilbert spaces has dimension 2, but it may not be the case for higher dimensions, as counterexamples shows.

For this particular article the maximal Choi ranks of UCPT maps will be important, as we'll see in theorem \ref{theorem: tensor product of ucpt need not be extreme}. In \cite{ohno} it was investigated the maximal Choi ranks of UCPT maps, but they were only computed up to dimension 4. For higher dimensions, there was given a lower bound that the maximal Choi rank is at least the dimension of the space. Also, for dimensions 3 and 4 were given single examples of such extreme UCPT maps of maximal rank. In \cite{Haagerup2021} was constructed families of examples of high rank extreme maps, but the lower bounds on the maximal ranks were not improved. A better lower bound on the maximal ranks of extreme UCPT maps may be enough to decide for which dimensions the tensor product preserves extremality for UCPT maps, but this wasn't done yet, as far as I know.

The article is organized as follows. In section 2 we review some concepts of Linear Algebra, Convex Sets and Quantum Information Theory. Also, we review the characterization of linear independence as the invertibility of the Gram matrix, which will be important in the cases of CPT and UCP maps. Then, in section 3, we apply these results to prove that the tensor product of extreme CPT or UCP maps is also extreme. Also, we show that if two extreme UCPT maps have high enough Choi ranks, their tensor product cannot be an extreme point of UCPT maps. Then we use examples of \cite{ohno} of high enough Choi ranks so that their tensor product is not extreme.

\section{Preliminaries}

\subsection{Linear Algebra}

\subsubsection{Hilbert spaces}

As usual in Quantum Mechanics, the inner product of a Hilbert space is antilinear in the first argument and linear in the second. Also, we restrict to complex finite dimensional vector spaces. Therefore, such a Hilbert space is a finite dimensional complex vector space $X$, together with an inner product $\langle \; , \; \rangle \colon X \times X \to \complexNumbers$. The inner product is a sesquilinear form, which means that the following properties are valid:
\begin{itemize}
	\item $\forall x,x',x''\in X$, $\langle x,x'+x''\rangle = \langle x,x'\rangle+ \langle x,x''\rangle$;
	\item $\forall x,x' \in X$, $\forall \lambda \in \complexNumbers$, $\langle x,\lambda x'\rangle = \lambda \langle x,x'\rangle$;
	\item $\forall x,x' \in X$, $\langle x,x'\rangle = \overline{\langle x',x\rangle}$.
\end{itemize} 

We denote by $\Hom(X,Y)$ the set of linear maps $A \colon X \to Y$. We define $\End(X) = \Hom(X,X)$, the set of linear endomorphisms of $X$. Also, $\Hom(X,Y)$ have the Hilbert-Schmidt inner product, given by $\langle A, B \rangle = tr(A^\dag B)$, for $A,B \in \Hom(X,Y)$.

A construction that will be important in this article is the tensor product. The tensor product of vector spaces $X$ and $Y$ is a $\complexNumbers$-bilinear map $\otimes \colon X\times Y \to X\otimes Y$ that satisfies the following universal property: if $B \colon X\times Y \to Z$ is $\complexNumbers$-bilinear, then there exists a unique $\complexNumbers$-linear map $u \colon X\otimes Y \to Z$ such that $b = u\otimes$.
\begin{center}
\begin{tikzcd}
X \times Y \arrow[r,"\otimes"] \arrow[rd,"B"] & X \otimes Y \arrow[d,dashed,"u"]\\
 & Z
\end{tikzcd}
\end{center}
Given a basis $x_1$,...,$x_n$ for $X$ and $y_1$,...,$y_m$ for $Y$, $(x_i \otimes y_j)_{\substack{i=1,...,n \\ j=1,...,m}}$ is a basis for $X \otimes Y$. We defined the tensor product of vector spaces, but there is also a tensor product of linear maps. Let $A \colon X \to X'$ and $B \colon Y \to Y'$ be $\complexNumbers$-linear maps. Their tensor product $A \otimes B \colon X \otimes Y \to X' \otimes Y'$ is the unique linear map such that the next diagram commutes:
\begin{center}
\begin{tikzcd}
X\times Y \arrow[d,"A\times B"] \arrow[r,"\otimes"] & X \otimes Y \arrow[d,dashed,"A\otimes B"]\\
X'\times Y' \arrow[r,"\otimes"] & X' \otimes Y'
\end{tikzcd}
\end{center}
That is, $A\otimes B$ is computed on elementary tensors as $(A\otimes B)(x \otimes y) = A(x) \otimes B(y)$, for any $x \in X$ and $y \in Y$. Given basis $x_1$,...,$x_n$ for $X$, $x'_1$,...,$x'_{n'}$ for $X'$, $y_1$,...,$y_m$ for $Y$ and $y'_1$,...,$y'_{m'}$ for $Y'$, we have the basis $(x_i \otimes y_j)_{\substack{i=1,...,n \\ j=1,...,m}}$ for $X \otimes Y$ and $(x'_{i'} \otimes y'_{j'})_{\substack{i'=1,...,n' \\ j'=1,...,m'}}$ for $X' \otimes Y'$. We can express the matrix $[A\otimes B]$ of $A\otimes B$ in these basis. Computing $A\otimes B$ on $x_i \otimes y_j$, we have $(A\otimes B)(x_i \otimes y_j) = A(x_i) \otimes B(y_j) = (\sum_{i'} [A]_{i',i} x'_{i'}) \otimes (\sum_{j'} [B]_{j',j} y'_{j'})$. By bilinearity of $\otimes$, we get $(A\otimes B)(x_i \otimes y_j) = \sum_{i',j'} [A]_{i',i} [B]_{j',j} x'_{i'} \otimes y'_{j'}$. Therefore, $[A\otimes B]$ has elements $[A\otimes B]_{(i',j'),(i,j)} = [A]_{i',i} [B]_{j',j}$. Since $[A \otimes B]$ has double indices $(i,j)$, to interpret it as a matrix we have to choose an order. This is done by taking the lexicographic order, that is, $(i,j) < (i',j') \iff i<i'$ or $i=i'$ and $j<j'$. With this order, the matrix $[A\otimes B]$ is called the Kronecker product of $[A]$ and $[B]$, which we denote as $[A] \otimes  [B]$.

If $X$ and $Y$ are Hilbert spaces, then their tensor product $X \otimes Y$ is also a Hilbert space, with inner product given by $\langle x \otimes y, x' \otimes y' \rangle = \langle x,x' \rangle \langle y,y' \rangle$, for $x,x' \in X$ and $y,y' \in Y$. We have a similar property for the Hilbert-Schmidt inner product on $\End(X\otimes Y)$. Any element of $\End(X\otimes Y)$ is a linear combination of linear maps of the form $A \otimes B$, for $A \in \End(X)$ and $B \in \End(Y)$. Therefore, they are sufficient to characterize the inner product on $\End(X\otimes Y)$. Given $A,A' \in \End(X)$ and $B,B' \in \End(Y)$, we have $\langle A \otimes B, A' \otimes B' \rangle = tr((A \otimes B)^\dag (A' \otimes B')) = tr((A^\dag \otimes B^\dag)(A' \otimes B')) = tr(A^\dag A' \otimes B^\dag B')$. The trace is multiplicative, that is, we have the property $tr(f \otimes g) = tr(f)tr(g)$. Using this property, we get $\langle A \otimes B, A' \otimes B' \rangle = tr(A^\dag A') tr(B^\dag B') = \langle A,A' \rangle \langle B,B' \rangle$. Therefore 
\begin{equation}
\langle A \otimes B, A' \otimes B' \rangle = \langle A,A' \rangle \langle B,B' \rangle,
\label{equation: HS multiplicativity}
\end{equation}
for any $A,A' \in \End(X)$ and $B,B' \in \End(Y)$.

\subsubsection{Gram matrices}

As we'll see in theorem \ref{theorem: characterization of extreme cpt maps}, it characterizes extremality of CPT maps in terms of linear independence. It will be useful to restate linear independence as the invertibility of the Gram matrix. In this section, we recall some results regarding Gram matrices. A proof of lemma \ref{lemma: gram matrix} may be found in theorem 7.2.10 of \cite{horn_johnson_2012}.

\begin{lemma}
	Let $(V,\langle \; , \; \rangle)$ be an inner product space and $v_1$,..., $v_n$ be vectors of $V$. Let
	\begin{equation}
		G = \begin{pmatrix}
			\langle v_1, v_1 \rangle & \langle v_1, v_2 \rangle & \cdots & \langle v_1, v_n \rangle \\
			\langle v_2, v_1 \rangle & \langle v_2, v_2 \rangle & \cdots & \langle v_2, v_n \rangle \\
			\vdots & \vdots & \ddots & \vdots\\
			\langle v_n, v_1 \rangle & \langle v_n, v_2 \rangle & \cdots & \langle v_n, v_n \rangle
		\end{pmatrix}
	\end{equation}
	be the Gram matrix of $v_1$,...,$v_n$. Then $v_1$,..., $v_n$ are linearly independent if, and only if, its Gram matrix is invertible.
	\label{lemma: gram matrix}
\end{lemma}
\begin{proof}
	First, given vectors $a = (a_1,...,a_n)$ and $b = (b_1,...,b_n)$, which we'll represent as column vectors, lets compute $a^\dag G b$. We have
	\begin{align*}
		Gb &= \begin{pmatrix}
			\langle v_1, v_1 \rangle & \langle v_1, v_2 \rangle & \cdots & \langle v_1, v_n \rangle \\
			\langle v_2, v_1 \rangle & \langle v_2, v_2 \rangle & \cdots & \langle v_2, v_n \rangle \\
			\vdots & \vdots & \ddots & \vdots\\
			\langle v_n, v_1 \rangle & \langle v_n, v_2 \rangle & \cdots & \langle v_n, v_n \rangle
		\end{pmatrix}	\begin{pmatrix}
			b_1\\
			b_2\\
			\vdots\\
			b_n
		\end{pmatrix}\\
		&= 	\begin{pmatrix}
			\sum_{j=1}^n \langle v_1, v_j \rangle b_j\\
			\sum_{j=1}^n \langle v_2, v_j \rangle b_j\\
			\vdots\\
			\sum_{j=1}^n \langle v_n, v_j \rangle b_j\\
		\end{pmatrix}\\
		&= 	\begin{pmatrix}
			\langle v_1,\sum_{j=1}^n b_j v_j \rangle\\
			\langle v_2,\sum_{j=1}^n b_j v_j \rangle\\
			\vdots\\
			\langle v_n,\sum_{j=1}^n b_j v_j \rangle\\
		\end{pmatrix}.
	\end{align*}
	Therefore $a^\dag G b = \sum_{i=1}^n \overline{a_i} \langle v_i,\sum_{j=1}^n b_j v_j \rangle = \langle \sum_{i=1}^n a_i v_i,\sum_{j=1}^n b_j v_j \rangle$. Notice that the inner product is antilinear in the first argument, so $a_i$ passes to inside the inner product without the conjugation. The last result implies that $G$ is a positive matrix, since $a^\dag G a = ||\sum_{i=1}^n a_i v_i ||^2$, which is $\geq 0$. Therefore, $G$ is diagonalizable and its eigenvalues are real and non-negative.
	
	Now we prove the equivalence of $v_1$,...,$v_n$ being linearly independent and $G$ being invertible. By the equation $a^\dag G a= ||\sum_{i=1}^n a_i v_i||^2$, the vectors $v_1$,..., $v_n$ are linearly independent if, and only if, $a=0$ is the unique solution to the equation $a^\dag G a=0$. Also notice that, since $G$ is diagonalizable, it is invertible if, and only if, $0$ isn't an eigenvalue of $G$. Let $c_1$,..., $c_n$ be orthonormal eigenvectors of $G$ and $\lambda_i$ the eigenvalue of $c_i$. Then $a$ is written as $a = \sum_{i=1}^n \alpha_i c_i$, for some $\alpha_i \in \complexNumbers$. Then we have $a^\dag G a = a^\dag G \sum_{i=1}^n \alpha_i c_i = \sum_{i=1}^n \alpha_i a^\dag G  c_i = \sum_{i=1}^n \alpha_i a^\dag \lambda_i  c_i = \sum_{i=1}^n \lambda_i \alpha_i a^\dag c_i = \sum_{i=1}^n \lambda_i \alpha_i \overline{\alpha_i} = \sum_{i=1}^n \lambda_i |\alpha_i|^2$. If each $\lambda_i$ isn't 0, then for any $a\neq 0$ we have $\sum_{i=1}^n \lambda_i |\alpha_i|^2 > 0$. Therefore, when each $\lambda_i$ isn't 0, $a^\dag G a$ can only be 0 if $a=0$. If for some $j$ we have $\lambda_j =0$, then we can pick $\alpha_j=1$ and $\alpha_i=0$ for $i\neq j$, which gives a vector $a\neq 0$ with $a^\dag G a = 0$. Therefore,  $a=0$ is the unique solution to the equation $a^\dag G a=0$ precisely when each $\lambda_i$ isn't 0, that is, when $G$ is invertible. This concludes that $G$ is invertible if, and only if, $v_1$,..., $v_n$ are linearly independent.
\end{proof}

Gram matrices of tensor products of vectors will be important. For this reason, we prove the following lemmas:

\begin{lemma}
Let $V,W$ be vectors spaces. Let also $(v_i)_{i=1,...,n}$ be vectors of $V$ and $(w_j)_{j=1,...,m}$ be vectors of $W$. Let $G$ be the Gram matrix of $(v_i)_{i=1,...,n}$, $G'$ be the Gram matrix of $(w_j)_{j=1,...,m}$ and $G''$ be the Gram matrix of $(v_i \otimes w_j)_{\substack{i=1,...,n \\ j=1,...,m}}$. We regard $(v_i \otimes w_j)_{\substack{i=1,...,n \\ j=1,...,m}}$ as a sequence ordered in the lexicographic order, that is, $(i,j) < (i',j') \iff i<i'$ or $i=i'$ and $j<j'$. Then $G''$ is the Kronecker product of $G$ and $G'$, that is, $G''= G \otimes G'$. Its matrix elements are, therefore, $G''_{(i,j),(i',j')} = G_{i,i'} G'_{j,j'}$.
\label{lemma: Gram tensor}
\end{lemma}
\begin{proof}
The matrix elements of $G$, $G'$ and $G''$ are, by definition,
\begin{equation}
G_{i,i'} = \langle v_i , v_{i'} \rangle,
\label{eq: G elemento de matriz}
\end{equation}
\begin{equation}
G'_{j,j'} = \langle w_j , w_{j'} \rangle, 
\label{eq: G' elemento de matriz}
\end{equation}
\begin{equation}
G''_{(i,j),(i',j')} = \langle v_i \otimes w_j , v_{i'} \otimes w_{j'} \rangle.
\end{equation}
By definition of the inner product of $V \otimes W$, we have $G''_{(i,j),(i',j')} $ $=$ $ \langle v_i \otimes w_j , v_{i'} \otimes w_{j'} \rangle $ $=$ $ \langle v_i, v_{i'}\rangle \langle w_j , w_{j'} \rangle$. By (\ref{eq: G elemento de matriz}) and (\ref{eq: G' elemento de matriz}), we have $G''_{(i,j),(i',j')} = G_{i,i'} G'_{j,j'}$. Since the indices $(i,j)$ have the lexicographic order, the last equation says that $G'' = G \otimes G'$.
\end{proof}

\begin{lemma}
	Let $(v_i)_i$ be a finite sequence in some vector space $V$ and let $(w_j)_j$ be a finite sequence in some vector space $W$. Then $(v_i \otimes w_j)_{i,j}$ is linearly independent if, and only if, both $(v_i)_i$ and $(w_j)_j$ are linearly independent.
	\label{lemma: tensor of l.i. vectors}
\end{lemma}
\begin{proof}
	We use the same notation of lemma \ref{lemma: Gram tensor}. Let $n$ be the size of $(v_i)_i$ and $m$ be the size of $(w_j)_j$. The sequences $(v_i)_i$, $(w_j)_j$ and $(v_i \otimes w_j)_{i,j}$ have Gram matrices $G$, $G'$ and $G''$, respectively. By lemma \ref{lemma: Gram tensor}, we have $G'' = G \otimes G'$. Also, $\det(G'') = \det(G \otimes G') = \det(G)^m \det(G')^n$, so $G''$ is invertible if, and only if, both $G$ and $G'$ are invertible. In this case, $G''^{-1} = G^{-1} \otimes G'^{-1}$. Using lemma \ref{lemma: gram matrix} we conclude that $(v_i \otimes w_j)_{i,j}$ is linearly independent if, and only if, both $(v_i)_i$ and $(w_j)_j$ are linearly independent.
\end{proof}

\subsection{Convex Sets}

Given a vector space $X$, a convex set $C \subseteq X$ is a subset $C$ that is closed by convex combinations. A convex combination of vectors $v_1$,...,$v_n$ is a linear combination $\sum_{i=1}^n p_i v_i$, where $p_i \geq 0$ and $\sum_{i=1}^n p_i = 1$. For any subset $S\subseteq X$, its convex hull is the set of convex combinations of points of $S$. Therefore, a subset $C\subseteq X$ is convex if it equals its convex hull. An extreme point of a convex $C$ is one that can't be written as a convex combination of different points of $C$. For example, $[0,1]$ is a convex with 0 and 1 as extreme points. By a theorem of Krein and Milman, every compact convex set of a finite dimensional vector space is the convex hull of its extreme points (theorem 3.3 of \cite{barvinokcourse}).

\begin{theorem}
	Let $X$ be a finite dimensional complex vector space and $C \subseteq X$ a compact convex set. Then $C$ is the convex hull of its extreme points.
\label{theorem: convex hull}
\end{theorem}

Of course, this is not true for general convex sets. For example, $[0,+\infty)$ is a convex set, but isn't compact. Also, it has 0 as its unique extreme point, whose convex hull is $\{0\}$, which isn't $[0,+\infty)$.

\subsection{CP maps}

Every map we'll consider in this article will be some special kind of completely positive (CP) map. Let $X$ and $Y$ be finite dimensional Hilbert spaces. A CP map $\varepsilon \colon X \to Y$ is a linear map $\varepsilon \colon \End(X) \to \End(Y)$ such that $\id_{\End(Z)} \otimes \varepsilon \colon \End(Z \otimes X) \to \End(Z \otimes Y)$ is a positive map, for any finite dimensional Hilbert space $Z$. A positive map is one that sends positive operators to positive operators. A positive operator is a diagonalizable operator whose eigenvalues are all real and non-negative. We'll denote as $CP(X,Y)$ the set of CP maps from $X$ to $Y$. Instead of using this definition of complete positivity, we'll often use the following characterization. A linear map $\varepsilon \colon \End(X) \to \End(Y)$ is CP if, and only if, it is of the form $\varepsilon(A) = \sum_k E_k A E_k^\dag$, for some finite sequence $(E_k)_k$ of operators $E_k \colon X \to Y$. The operators $E_k$ are called operation elements or Kraus operators. Note that the sequence $(E_k)_k$ that represents a CP map is not unique. The minimum number of operators $E_k$ necessary to represent $\varepsilon$ is called its Choi rank, which we'll denote by $CR(\varepsilon)$. The number of indices $k$ in the sequence $(E_k)_k$ equals the Choi rank if, and only if, $(E_k)_k$ is linearly independent. For a proof of these statements, see \cite{CHOI1975285} or \cite{watrous_2018}.

There is a duality between CP maps that will be useful to relate CPT and UCP maps. Since $\End(X)$ and $\End(Y)$ are Hilbert spaces with the Hilbert-Schmidt inner product, we have a notion of dual map $\varepsilon^\dag$. It is defined as the unique linear map $\varepsilon^\dag \colon \End(Y) \to \End(X)$ that satisfies $\langle \varepsilon(A), B \rangle = \langle A, \varepsilon^\dag(B) \rangle$, for every $A \in \End(X)$ and $B \in \End(Y)$.

\begin{lemma}
	A linear map $\varepsilon\colon \End(X) \to \End(Y)$ is CP if, and only if, its dual $\varepsilon^\dag \colon \End(Y) \to \End(X)$ is CP. Also, if $\varepsilon$ has operation elements $(E_k)_k$, then $(E_k^\dag)_k$ are operation elements for $\varepsilon^\dag$.
	\label{lemma: duality of cps}
\end{lemma}
\begin{proof}
	Suppose that $\varepsilon$ is a completely positive map with operation elements $(E_k)_k$. For any $A \in \End(X)$ and $B \in \End(Y)$, we have
	\begin{align*}
		\langle A, \varepsilon^\dag(B) \rangle &= \langle \varepsilon(A), B \rangle\\
		&=  \langle \sum_k E_k A E_k^\dag, B \rangle\\
		&= \sum_k \langle E_k A E_k^\dag, B \rangle\\
		&= \sum_k tr((E_k A E_k^\dag)^\dag B)\\
		&= \sum_k tr(E_k A^\dag E_k^\dag B)\\
		&= \sum_k tr(A^\dag E_k^\dag B E_k)\\
		&= tr(A^\dag \sum_k E_k^\dag B E_k)\\
		&= \langle A, \sum_k E_k^\dag B E_k \rangle.
	\end{align*}
	Since $\varepsilon^\dag$ is unique, we conclude that $\varepsilon^\dag (B) = \sum_k E_k^\dag B E_k$. Therefore, $\varepsilon^\dag$ is completely positive and has $(E_k^\dag)_k$ as operation elements.
	
	If we assume instead that $\varepsilon^\dag$ is completely positive, using that $\varepsilon^{\dag\dag} = \varepsilon$ we conclude, by the previous results, that $\varepsilon$ is completely positive. Therefore, $\varepsilon$ is completely positive if, and only if, $\varepsilon^\dag$ is completely positive.
\end{proof}

The tensor product will be a central topic of this article, so we review some results about it. For any CP maps $\varepsilon \in CP(X,Y)$ and $\varepsilon' \in CP(X',Y')$ there is the tensor product $\varepsilon \otimes \varepsilon'$. We'll see that $\varepsilon \in CP(X \otimes X', Y \otimes Y')$. Also, given operation elements for $\varepsilon$ and $\varepsilon'$, we can construct operation elements for $\varepsilon \otimes \varepsilon'$ as in the next lemma.

\begin{lemma}
	If $\varepsilon \in CP(X,Y)$ and $\varepsilon' \in CP(X',Y')$ have operation elements $(E_k)_k$ and $(F_l)_l$, respectively, then $\varepsilon \otimes \varepsilon'$ have $(E_k \otimes F_l)_{k,l}$ as operation elements. In particular, $\varepsilon \in CP(X \otimes X', Y \otimes Y')$.
	\label{lemma: tensor product operation elements}
\end{lemma}
\begin{proof}
	$\varepsilon\otimes \varepsilon'$ is characterized as the map that satisfies $(\varepsilon \otimes \varepsilon')(A\otimes B) = \varepsilon(A) \otimes \varepsilon'(B)$, for any $A \in \End(X)$ and $B \in \End(X')$. Therefore $(\varepsilon \otimes \varepsilon')(A\otimes B) = \varepsilon(A) \otimes \varepsilon'(B) = (\sum_k E_k A E_k^\dag) \otimes (\sum_l F_l B F_l^\dag) = \sum_{k,l} E_k A E_k^\dag \otimes \sum_l F_l B F_l^\dag \stackrel{(*)}{=} \sum_{k,l} (E_k \otimes F_l)(A\otimes B) (E_k^\dag \otimes F_l^\dag) = \sum_{k,l} (E_k \otimes F_l)(A\otimes B) (E_k \otimes F_l)^\dag$. In the equality ($*$) we used the identity $(f' \otimes g')(f \otimes g) = f'f \otimes g'g$. The identity $(\varepsilon \otimes \varepsilon')(A\otimes B) = \sum_{k,l} (E_k \otimes F_l)(A\otimes B)(E_k \otimes F_l)^\dag$ together with the linearity of $\varepsilon \otimes \varepsilon'$ implies that $(\varepsilon \otimes \varepsilon')(C) = \sum_{k,l} (E_k \otimes F_l)C(E_k \otimes F_l)^\dag$, for any $C \in \End(X \otimes X')$. Therefore $(E_k \otimes F_l)_{k,l}$ are operation elements for $\varepsilon \otimes \varepsilon'$. Since $\varepsilon \otimes \varepsilon'$ is of the form $(\varepsilon \otimes \varepsilon')(C) = \sum_{k,l} (E_k \otimes F_l)C(E_k \otimes F_l)^\dag$, it is completely positive.
\end{proof}

\begin{corollary}
	For any CP maps $\varepsilon$ and $\varepsilon'$ we have $CR(\varepsilon \otimes \varepsilon') = CR(\varepsilon) CR(\varepsilon')$.
\label{corollary: CR of tensor product}
\end{corollary}
\begin{proof}
	Let $(E_k)_k$ be linearly independent operation elements for $\varepsilon$ and let $(F_l)_l$ be linearly independent operation elements for $\varepsilon'$. By lemmas \ref{lemma: tensor of l.i. vectors} and \ref{lemma: tensor product operation elements} we know that $(E_k \otimes F_l)_{k,l}$ are linearly independent operation elements for $\varepsilon \otimes \varepsilon'$. Since the operations elements are linearly independent, the sequences $(E_k)_k$ and $(F_l)_l$ have $CR(\varepsilon)$ and $CR(\varepsilon')$ operators, respectively. Therefore $(E_k \otimes F_l)_{k,l}$ is a sequence with $CR(\varepsilon) CR(\varepsilon')$ operators. But $(E_k \otimes F_l)_{k,l}$ is linearly independent, so it must have $CR(\varepsilon \otimes \varepsilon')$ operators. Therefore $CR(\varepsilon \otimes \varepsilon') = CR(\varepsilon) CR(\varepsilon')$.
\end{proof}

For later use, we'll show that $(\varepsilon \otimes \varepsilon')^\dag = \varepsilon^\dag \otimes \varepsilon'^\dag$.

\begin{lemma}
	Let $\varepsilon$ and $\varepsilon'$ be completely positive maps. We have that $(\varepsilon \otimes \varepsilon')^\dag = \varepsilon^\dag \otimes \varepsilon'^\dag$.
	\label{lemma: dual of tensor of cp maps}
\end{lemma}
\begin{proof}
	Lets write that $\varepsilon \in CP(X,Y)$ and $\varepsilon' \in CP(X',Y')$. Both $(\varepsilon \otimes \varepsilon')^\dag$ and $\varepsilon^\dag \otimes \varepsilon'^\dag$ are elements of $CP(Y\otimes Y', X \otimes X')$. Their domain is $\End(Y\otimes Y')$, which is spanned by the maps $A \otimes A'$, for $A \in \End(Y)$ and $A' \in \End(Y')$. Similarly for the codomain $\End(X \otimes X')$, it is spanned by the maps $B \otimes B'$ for $B \in \End(X)$ and $B' \in \End(X')$. By linearity, we'll only need to consider these $A\otimes A'$ and $B\otimes B'$. By definition of the dual, $(\varepsilon \otimes \varepsilon')^\dag$ is the unique linear map that satisfies $\langle (\varepsilon \otimes \varepsilon')(A \otimes A') , B \otimes B' \rangle = \langle A \otimes A', (\varepsilon \otimes \varepsilon')^\dag(B \otimes B') \rangle$. We have that
	\begin{align*}
		\langle (\varepsilon \otimes \varepsilon')(A \otimes A') , B \otimes B' \rangle &= \langle \varepsilon(A) \otimes \varepsilon'(A') , B \otimes B' \rangle\\
		&= \langle \varepsilon(A), B \rangle \langle \varepsilon'(A'), B' \rangle\\
		&= \langle A, \varepsilon^\dag(B) \rangle \langle A', \varepsilon^\dag(B') \rangle\\
		&= \langle A \otimes A', \varepsilon(B) \otimes \varepsilon'(B') \rangle\\
		&= \langle A \otimes A', (\varepsilon \otimes \varepsilon')(B \otimes B') \rangle.
	\end{align*}
	By uniqueness it follows that $(\varepsilon \otimes \varepsilon')^\dag = \varepsilon^\dag \otimes \varepsilon'^\dag$.
\end{proof}

\subsubsection{CPT maps}

A CPT map $\varepsilon \colon X \to Y$ is an $\varepsilon \in CP(X,Y)$ which is also trace preserving. Being trace preserving means that $tr(\varepsilon(A)) = tr A$, for every $A\in \End(X)$. This is equivalent to the equation $\sum_k E_k^\dag E_k = id_X$. Lets denote by $CPT(X,Y)$ the set of CPT maps $\varepsilon \colon X \to Y$.

As the two next propositions shows, $CPT(X,Y)$ is convex and compact.

\begin{prop}
$CPT(X,Y)$ is a convex set. 
\end{prop}
\begin{proof}

Let $p_1,...,p_n$ be numbers satisfying $p_i \geq 0$ and $\sum_{i=1}^n p_i=1$. Also, let $\varepsilon_1$,...,$\varepsilon_n$ be CPT maps in $CPT(X,Y)$. Let $(E_{i,k_i})_{k_i}$ be operation elements for $\varepsilon_i$, so that $\varepsilon_i(A) = \sum_{k_i} E_{i,k_i} A E_{i,k_i}^\dag$. Then the convex combination of $\varepsilon_i$ is given by $(\sum_{i=1}^n p_i \varepsilon_i)(A) $ $=$ $ \sum_{i=1}^n p_i \sum_{k_i} E_{i,k_i}A E_{i,k_i}^\dag $ $=$ $ \sum_{i=1}^n \sum_{k_i}$ $ (\sqrt{p_i}E_{i,k_i})A $ $(\sqrt{p_i}E_{i,k_i})^\dag$. This shows that $\sum_{i=1}^n p_i \varepsilon_i$ is completely positive and has operation elements $(\sqrt{p_i} E_{i,k_i})_{i,k_i}$. Also, it is trace preserving, since $\sum_{i,k_i} (\sqrt{p_i} E_{i,k_i})^\dag $ $ (\sqrt{p_i} E_{i,k_i})$ $=$ $ \sum_i p_i \sum_{k_i} E_{i,k_i}^\dag$ $ E_{i,k_i} $ $=$ $\sum_i p_i id_X = id_X$.
\end{proof}

For the compacity we use the \CJ isomorphism \cite{watrous_2018}. There is a well known basis dependent version of this isomorphism, but here we use a basis independent one. Given an orthonormal basis $x_1$, ..., $x_n$ of $X$, let $x^1$,...,$x^n$ denote the dual basis for $X^*$. That is, $x^i \in X^*$ and $x^i(x_j) = \delta_{i,j}$. Let $\Omega = \sum_{i=1}^n x^i \otimes x_i \in X^* \otimes X$.  The \CJ isomorphism is the linear isomorphism $CJ \colon$ $ \Hom(\End(X),\End(Y)) $ $\to$ $\End(X^* \otimes Y)$ given by $CJ(\varepsilon) = (id_{\End(X^*)} \otimes \varepsilon)(|\Omega \rangle \langle \Omega |)$ (in Dirac notation). When restricted to $CPT(X,Y)$, it gives a bijection $CJ \colon CPT(X,Y) \to \{A \in \End(X^* \otimes Y) \mid A \text{ is positive, } tr_Y A = id_{X^*}\}$. The well known basis dependent version differs by replacing $x^i$ with $x_i$ and $X^*$ with $X$ in the previous definitions.

\begin{prop}
$CPT(X,Y)$ is compact, in the norm topology.
\end{prop}
\begin{proof}
By the \CJ isomorphism, $CPT(X,Y)$ is homeomorphic to $CJ(X,Y)\coloneqq \{A \in \End(X^* \otimes Y) \mid A \text{ is positive, } tr_Y A = id_{X^*}\}$. Since for finite dimensions every norm is equivalent, it suffices to show that $CJ(X,Y)$ is bounded and closed, with respect to any norm on $CJ(X,Y)$. We show this for the operator norm.

First we show that $CJ(X,Y)$ is bounded. For any $A \in CJ(X,Y)$ we have $tr_Y A = id_{X^*}$, therefore $tr A = tr_{X^*} tr_Y A = tr_{X^*} id_{X^*} = \dim X^* = \dim X$. Since $A$ is also positive, it is diagonalizable and has non negative eigenvalues. Let $Sp(A)$ be the spectrum of $A$, that is, its set of eigenvalues. Then $||A|| = max_{\lambda \in Sp(A)} |\lambda| \stackrel{(*)}{=} max_{\lambda \in Sp(A)} \lambda = max (Sp(A))$. In $(*)$ we've used that $A$ has non negative eigenvalues. Since $tr A$ is the sum of the eigenvalues of $A$, we have $\lambda \leq tr A \leq dim X$, for any $\lambda \in Sp(A)$. Therefore $||A|| \leq dim X$. This concludes that $CJ(X,Y)$ is bounded.

Now lets show that $CJ(X,Y)$ is closed. Let $(A_n)_{n \in \naturalNumbers}$ be a sequence in $CJ(X,Y)$ that converges to some operator $A \in \End(X^* \otimes Y)$. We can easily show that $A$ is positive. Let $v \in X^* \otimes Y$, then $\langle v, A(v) \rangle = \langle v, \lim_n  A_n(v) \rangle$. Since the inner product is continuous in both its entries, then $\langle v, \lim_n  A_n(v) \rangle = \lim_n \langle v, A_n(v) \rangle$. Since each $A_n$ is positive, we have $\langle v, A_n(v) \rangle \geq 0$, so $\lim_n \langle v, $ $A_n(v) \rangle$ $ \geq$ $0$. Therefore $\langle v, A(v) \rangle \geq 0$, for any $v \in X^* \otimes Y$, which implies that $A$ is positive. Also, since $tr_Y$ is linear, it is continuous. Therefore, $tr_Y A = tr_Y (\lim_n A_n) = \lim_n tr_Y A_n$. Since $tr_Y A_n = id_{X^*}$, we have $\lim_n tr_Y A_n = \lim_n id_{X^*} = id_{X^*}$. Therefore, $tr_Y A = id_X$ and we conclude that $A \in CJ(X,Y)$. This proves that $CJ(X,Y)$ is closed.

Since $CJ(X,Y)$ is bounded and closed, it is compact. Using the \CJ isomorphism we conclude that $CPT(X,Y)$ is also compact.

\end{proof}

Being a compact convex set, by theorem \ref{theorem: convex hull} we know that $CPT(X,Y)$ is the convex hull of its extreme points. 

In \cite{watrous_2018} is given the following characterization of extremality of a CPT map (theorem 2.31 of \cite{watrous_2018}):

\begin{theorem}
	Let $\varepsilon \in CPT(X,Y)$ be a CPT map with linearly independent operation elements $(E_k)_k$. Then $\varepsilon$ is an extreme point of $CPT(X,Y)$ if, and only if, the family $(E_k^\dag E_{k'})_{k,k'}$ is linearly independent.
\label{theorem: characterization of extreme cpt maps}
\end{theorem}

Notice that we can always find operation elements for $\varepsilon$ which are linearly independent. This can be achieved by applying the \CJ isomorphism to $\varepsilon$ and computing an orthonormal basis that diagonalizes $CJ(\varepsilon)$. Then the $E_k$'s can be constructed from the eigenvectors with non zero eigenvalue.

In lemma \ref{lemma: tensor product operation elements} we've shown that $CP$ maps are closed under the tensor product. For completeness, lets show the same result for CPT maps.

\begin{lemma}
	If $\varepsilon \in CPT(X,Y)$ and $\varepsilon' \in CPT(X',Y')$ then $\varepsilon \otimes \varepsilon' \in CPT(X \otimes X',Y \otimes Y')$.
	\label{lemma: tensor product of cpt}
\end{lemma}
\begin{proof}
	By lemma \ref{lemma: tensor product operation elements} we already know that $\varepsilon \otimes \varepsilon' \in CP(X \otimes X',Y \otimes Y')$. It remains to check that $\varepsilon \otimes \varepsilon'$ is trace preserving. Lets use the same notation of lemma \ref{lemma: tensor product operation elements}. Being trace preserving means that $\sum_{k,l} (E_k \otimes F_l)^\dag (E_k \otimes F_l) = id_{X \otimes X'}$. We prove it as follows: $\sum_{k,l} (E_k \otimes F_l)^\dag (E_k \otimes F_l) = \sum_{k,l} (E_k^\dag \otimes F_l^\dag) (E_k \otimes F_l) = \sum_{k,l} E_k^\dag E_k \otimes F_l^\dag F_l =  (\sum_k E_k^\dag E_k) \otimes (\sum_l F_l^\dag F_l) = id_X \otimes id_{X'} = id_{X \otimes X'}$.
\end{proof}

\subsubsection{UCP maps}

Unital completely positive (UCP) maps are CP maps $\varepsilon \in CP(X,Y)$ that preserve the identity, that is, $\varepsilon(id_X) = id_Y$. We denote the set of UCP maps from $X$ to $Y$ by $UCP(X,Y)$. If $\varepsilon$ has operation elements $(E_k)_k$, being unital is equivalent to the equation $\sum_k E_k E_k^\dag = id_Y$. We can see that this equation is similar to $\sum_k E_k^\dag E_k = id_X$, which is equivalent to being trace preserving. Such similarity can be seen as a consequence of the fact that UCP maps are dual to CPT maps. 

\begin{prop}
	A linear map $\varepsilon \colon \End(X) \to \End(Y)$ is CPT if, and only if, $\varepsilon^\dag \colon \End(Y) \to \End(X)$ is UCP. In particular, we have an anti-linear bijection $\dag \colon CPT(X,Y) \to UCP(Y,X)$.
\label{prop: duality cpt ucp}
\end{prop}
\begin{proof}
	Assume that $\varepsilon$ is completely positive with operation elements $(E_k)_k$. By lemma \ref{lemma: duality of cps} we know that $\varepsilon^\dag$ is also completely positive and has $(E_k^\dag)_k$ as operation elements. Then $\varepsilon$ is trace preserving if $\sum_k E_k^\dag E_k = id_X$, and $\varepsilon^\dag$ is unital if $\sum_k E_k^\dag (E_k^\dag)^\dag = id_X$. Both are the same equation, therefore $\varepsilon$ is trace preserving if, and only if, $\varepsilon^\dag$ is unital.
\end{proof}

We can use the duality between $CPT$ and $UCP$ maps to transfer results between them.

\begin{prop}
	UCP(X,Y) is compact convex, in the norm topology.
\end{prop}
\begin{proof}
	By proposition \ref{prop: duality cpt ucp}, we have an anti-linear bijection $\dag \colon CPT(X,Y) \to UCP(Y,X)$, for any finite dimensional Hilbert spaces $X$ and $Y$. Since it is anti-linear, it is also continuous. Then $UCP(X,Y)$  is the image of the compact $CPT(Y,X)$ by the continuous function $\dag$, and therefore $UCP(X,Y)$ is compact. Also, if $\varepsilon_i \in UCP(X,Y)$ and $(p_i)_i$ is a discrete probability distribution, then $\sum_i p_i \varepsilon_i = (\sum_i p_i \varepsilon_i)^{\dag \dag} =  (\sum_i p_i( \varepsilon_i)^\dag)^\dag$. Sincd $CPT(Y,X)$ is convex, we know that $\sum_i p_i( \varepsilon_i)^\dag \in CPT(Y,X)$, and therefore $ (\sum_i p_i( \varepsilon_i)^\dag)^\dag \in UCP(X,Y)$. That is, $\sum_i p_i \varepsilon_i \in UCP(X,Y)$, which implies that $UCP(X,Y)$ is convex.
\end{proof}

The anti-linear bijection $\dag \colon CPT(X,Y) \to UCP(Y,X)$ stablishes a bijection between the extreme points of $CPT(X,Y)$ and the extreme points of $UCP(Y,X)$. Using this duality, we obtain a characterization of extreme points for $UCP(X,Y)$ that is similar to theorem \ref{theorem: characterization of extreme cpt maps}.

\begin{theorem}
	Let $\varepsilon \in UCP(X,Y)$ be a UCP map with linearly independent operation elements $(E_k)_k$. Then it is an extreme point of $UCP(X,Y)$ if, and only if, $(E_k E_l^\dag)_{k,l}$ is linearly independent.
\end{theorem}
\begin{proof}
	$\varepsilon$ is an extreme point of $UCP(X,Y)$ if, and only if, $\varepsilon^\dag$ is an extreme point of $CPT(Y,X)$. By lemma \ref{lemma: duality of cps}, $(E_k^\dag)_k$ are operation elements for $\varepsilon$. Since $(E_k)_k$ is linearly independent, it follows that $(E_k^\dag)_k$ is also linearly independent. Then, by theorem \ref{theorem: characterization of extreme cpt maps}, $\varepsilon^\dag$ is an extreme point of $CPT(Y,X)$ if, and only if, $(E_k^{\dag\dag} E_l^\dag)_{k,l}$ is linearly independent. But $E_k^{\dag\dag} = E_k$, so $\varepsilon$ is an extreme point of $UCP(X,Y)$ if, and only if, $(E_k E_l^\dag)_{k,l}$ is linearly independent.
\end{proof}

Finally, lets prove that UCP maps are closed under the tensor product.

\begin{lemma}
	If $\varepsilon \in UCP(X,Y)$ and $\varepsilon' \in UCP(X',Y')$ then $\varepsilon \otimes \varepsilon' \in UCP(X \otimes X',Y \otimes Y')$.
	\label{lemma: tensor product of ucp}
\end{lemma}
\begin{proof}
	We can prove this result by using the duality between UCP and CPT maps and lemma \ref{lemma: tensor product of cpt}. Using lemma \ref{lemma: dual of tensor of cp maps} we have that $\varepsilon \otimes \varepsilon' = (\varepsilon^\dag \otimes \varepsilon'^\dag)^\dag$. We know that $\varepsilon \in UCP(X,Y)$ and $\varepsilon' \in UCP(X',Y')$, then proposition \ref{prop: duality cpt ucp} implies that $\varepsilon^\dag \in CPT(Y,X)$ and $\varepsilon'^\dag \in CPT(Y',X')$. By lemma \ref{lemma: tensor product of cpt} we know that $\varepsilon^\dag \otimes \varepsilon'^\dag \in CPT(Y \otimes Y', X \otimes X')$. Using proposition \ref{prop: duality cpt ucp} again we conclude that $(\varepsilon^\dag \otimes \varepsilon'^\dag)^\dag \in  UCP(X \otimes X',Y \otimes Y')$.
\end{proof}

\subsubsection{UCPT maps}

An UCPT map is one that is both unital and trace preserving, that is, $UCPT(X,$ $Y) $ $=$ $ UCP(X,Y) \cap CPT(X,Y)$. In this case, we restrict attention to endomorphisms ($X=Y$), since $X$ and $Y$ have necessarily the same dimension. In fact, let $\varepsilon \in UCPT(X,Y)$, then $\varepsilon (id_X) = id_Y$, because it is unital. But it is also trace preserving, so $tr(\varepsilon(id_X)) = tr(id_X)$, so $tr(id_Y) = tr(id_X)$. But $tr(id_X) = \dim X$, so $\dim X = \dim Y$. We'll denote $UCPT(X,X)$ as $UCPT(X)$.

\begin{prop}
	$UCPT(X)$ is compact convex, in the norm topology.
\end{prop}
\begin{proof}
	Since both $UCP(X)$ and $CPT(X)$ are convex sets, its intersection \linebreak $UCPT(X)$ is also convex. Since, both $UCP(X)$ and $CPT(X)$ are bounded sets, $UCPT(X)$ is also bounded. And since both $UCP(X)$ and $CPT(X)$ are closed sets, $UCPT(X)$ is also closed. Therefore $UCPT(X)$ is closed and bounded, so it is compact. 
\end{proof}

As $UCPT(X)$ is compact convex, it is the convex hull of its extreme points. The following characterization of the extreme points of $UCPT(X)$ was proven is \cite{LANDAU1993107}. It can also be found in theorem 4.21 of \cite{watrous_2018}.

\begin{theorem}
	Let $\varepsilon \in UCPT(X)$ be an UCPT map with linearly independent operation elements $(E_k)_k$. Then it is an extreme point of $UCPT(X)$ if, and only if, $(E_k^\dag E_l \oplus E_l E_k^\dag)_{k,l}$ is linearly independent as vectors of $\End(X) \oplus \End(X)$.
\label{theorem: characterization of extreme points for ucpt}
\end{theorem}

As done for the other types of CP maps, lets prove that UCPT is closed under the tensor product.

\begin{lemma}
	If $\varepsilon \in UCPT(X,Y)$ and $\varepsilon' \in UCPT(X',Y')$ then $\varepsilon \otimes \varepsilon' \in UCPT(X \otimes X',Y \otimes Y')$.
	\label{lemma: tensor product of ucpt}
\end{lemma}
\begin{proof}
	By lemmas \ref{lemma: tensor product of cpt} and \ref{lemma: tensor product of ucp} we know that $\varepsilon \otimes \varepsilon' \in CPT(X \otimes X',Y \otimes Y')$ and $\varepsilon \otimes \varepsilon' \in UCP(X \otimes X',Y \otimes Y')$. Since $UCPT(X \otimes X',Y \otimes Y') =  CPT(X \otimes X',Y \otimes Y') \cap UCP(X \otimes X',Y \otimes Y')$, this concludes the proof.
\end{proof}

\section{Main results}

\subsection{Extremality is preserved for CPT maps}

Now we prove that the tensor product preserves extramality for CPT maps.

\begin{theorem}
	If $\varepsilon$ is an extreme point of $CPT(X,Y)$ and $\varepsilon'$ is an extreme point of $CPT(X',Y')$, then $\varepsilon \otimes \varepsilon'$ is an extreme point of $CPT(X \otimes X', Y\otimes Y')$.
\label{theorem: extremality is preserved for cpt}
\end{theorem}
\begin{proof}
	We always use the Hilbert-Schmidt inner product. Let $(E_k)_k$ be linearly independent operation elements for $\varepsilon$, and let $(F_l)_l$ be linearly independent operation elements for $\varepsilon'$. By lemma \ref{lemma: tensor product operation elements}, $(E_k \otimes F_l)_{k,l}$ are operation elements for $\varepsilon\otimes \varepsilon'$. Also, lemma \ref{lemma: tensor of l.i. vectors} implies that $(E_k \otimes F_l)_{k,l}$ are linearly independent. Then, by theorem \ref{theorem: characterization of extreme cpt maps}, $\varepsilon \otimes \varepsilon'$ is extreme if, and only if, $((E_k \otimes F_l)^\dag (E_{k'}\otimes F_{l'}))_{k,k',l,l'}$ is linearly independent (as vectors of $\End(X \otimes Y)$). By lemma \ref{lemma: gram matrix}, linear independence is equivalent to the invertibility of the Gram matrix. To prove the invertibility of the Gram matrix, we'll show that it is the Kronecker product of two invertible matrices. Let $G$, $G'$ and $G''$ be the Gram matrices of $(E_k^\dag E_{k'})_{k,k'}$, $(F_l^\dag F_{l'})_{l,l'}$ and  $((E_k \otimes F_l)^\dag (E_{k'}\otimes F_{l'}))_{k,k',l,l'}$, respectively. Its matrix elements are
	\begin{equation}
		G_{(k,k'),(k'',k''')} = \langle E_k^\dag E_{k'}, E_{k''}^\dag E_{k'''} \rangle,
	\end{equation}
	\begin{equation}
		G'_{(l,l'),(l'',l''')} = \langle F_l^\dag F_{l'}, F_{l''}^\dag F_{l'''} \rangle,
	\end{equation}
	\begin{equation}
		G''_{(k,k',l,l'),(k'',k''',l'',l''')} = \langle (E_k \otimes F_l)^\dag (E_{k'}\otimes F_{l'}), (E_{k''}\otimes F_{l''})^\dag (E_{k'''}\otimes F_{l'''}) \rangle.
	\end{equation}
	Since $\varepsilon$ and $\varepsilon'$ are extreme CPT maps, and since $(E_k)_k$ and $(F_l)_l$ are linearly independent operation elements, by theorem \ref{theorem: characterization of extreme cpt maps} we know that $(E_k^\dag E_{k'})_{k,k'}$ and $(F_l^\dag F_{l'})_{l,l'}$ are both linearly independent. By lemma \ref{lemma: gram matrix}, it implies that $G$ and $G'$ are invertible matrices. We can rewrite $G''$ in terms of $G$ and $G'$ as follows:
	\begin{align*}
		G''_{(k,k',l,l'),(k'',k''',l'',l''')} &= \langle (E_k \otimes F_l)^\dag (E_{k'}\otimes F_{l'}), (E_{k''}\otimes F_{l''})^\dag (E_{k'''}\otimes F_{l'''}) \rangle\\
		&= \langle (E_k^\dag \otimes F_l^\dag)(E_{k'}\otimes F_{l'}), (E_{k''}^\dag\otimes F_{l''}^\dag) (E_{k'''}\otimes F_{l'''}) \rangle\\
		&= \langle E_k^\dag E_{k'} \otimes F_l^\dag F_{l'}, E_{k''}^\dag E_{k'''} \otimes F_{l''}^\dag F_{l'''} \rangle\\
		&\stackrel{(*)}{=} \langle E_k^\dag E_{k'}, E_{k''}^\dag E_{k'''} \rangle \langle F_l^\dag F_{l'}, F_{l''}^\dag F_{l'''} \rangle\\
		&= G_{(k,k'),(k'',k''')} G'_{(l,l'),(l'',l''')}.
	\end{align*}
	In equality (*) we've used equation (\ref{equation: HS multiplicativity}). Given any order of $(k,k')$ and $(l,l')$, we can give the lexicographic order to $(k,k',l,l')$. In this way, the identity $G''_{(k,k',l,l'),(k'',k''',l'',l''')} = G_{(k,k'),(k'',k''')} G'_{(l,l'),(l'',l''')}$ says that $G'' = G \otimes G'$, that is, $G''$ is the Kronecker product of $G$ and $G'$. Since $G$ and $G'$ are invertible, then $G''$ is also invertible, with inverse $G''^{-1} = G^{-1} \otimes G'^{-1}$. Then lemma \ref{lemma: gram matrix} implies that $((E_k \otimes F_l)^\dag (E_{k'} \otimes F_{l'}))_{k,k',l,l'}$ is linearly independent. Therefore, by theorem \ref{theorem: characterization of extreme cpt maps}, $\varepsilon \otimes \varepsilon'$ is an extreme CPT map.
	
\end{proof}

\subsection{Extremality is preserved for UCP maps}

The result for UCP maps can be obtained by using the duality between CPT and UCP maps.

\begin{theorem}
	If $\varepsilon$ is an extreme point of $UCP(X,Y)$ and $\varepsilon'$ is an extreme point of $UCP(X',Y')$, then $\varepsilon \otimes \varepsilon'$ is an extreme point of $UCP(X \otimes X', Y\otimes Y')$.
\end{theorem}
\begin{proof}
	Since $\varepsilon \in UCP(X,Y)$ and $\varepsilon' \in UCP(X',Y')$, we know that $\varepsilon^\dag \in CPT(Y,X)$ and $\varepsilon'^\dag \in CPT(Y',X')$. Also, $\varepsilon$ and $\varepsilon'$ are extreme points of $UCP(X,Y)$ and $UCP(X',Y')$, so $\varepsilon^\dag$ and $\varepsilon'^\dag$ are extreme points of $CPT(Y,X)$ and $CPT(Y',X')$. By theorem \ref{theorem: extremality is preserved for cpt}, we know that $\varepsilon^\dag \otimes \varepsilon'^\dag$ is an extreme point of $CPT(Y\otimes Y', X \otimes X')$. By lemma \ref{lemma: dual of tensor of cp maps} we know that $\varepsilon^\dag \otimes \varepsilon'^\dag = (\varepsilon \otimes \varepsilon')^\dag$, so $(\varepsilon \otimes \varepsilon')^\dag$ is extreme. Using again the duality between CPT and UCP maps we conclude that $\varepsilon \otimes \varepsilon'$ is an extreme point of $UCP(X \otimes X', Y \otimes Y')$.
\end{proof}

\subsection{The case of UCPT maps}

The tensor product of extreme UCPT maps may not be extreme. But it is always extreme if one of the two maps is over a Hilbert space of dimension 2. This is a consequence of the next proposition and the fact that, for dimension 2, extreme maps are unitary.

\begin{prop}
	Let $\varepsilon \in UCPT(X)$ and $\mathcal{U} \in UCPT(Y)$ be a unitary map, that is, $\mathcal{U}(A) = U A U^\dag$  for some unitary operator $U \in \End(Y)$. We assume that $\dim Y > 0$. Then $\varepsilon$ is an extreme point of $UCPT(X)$ if, and only if, $\varepsilon\otimes \mathcal{U}$ is an extreme point of $UCPT(X \otimes Y)$.
\label{prop: tensor with unitary}
\end{prop}
\begin{proof}
	Let $(E_k)_k$ be linearly independent operation elements for $\varepsilon$. $\mathcal{U}$ has already $(U)$ as linearly independent operation elements, since there is only one operator. By lemmas \ref{lemma: tensor of l.i. vectors} and \ref{lemma: tensor product operation elements}, $\varepsilon \otimes \mathcal{U}$ has linearly independent operation elements $(E_k \otimes U)$. By theorem \ref{theorem: characterization of extreme points for ucpt}, $\varepsilon \otimes \mathcal{U}$ is an extreme point of $UCPT(X\otimes Y)$ if, and only if, $((E_k \otimes U)^\dag (E_l \otimes U) \oplus (E_l \otimes U)(E_k \otimes U)^\dag)_{k,l}$ is linearly independent. We have that
	\begin{align*}
	 (E_k \otimes U)^\dag (E_l \otimes U) \oplus (E_l \otimes U)(E_k \otimes U)^\dag &= (E_k^\dag \otimes U^\dag) (E_l \otimes U) \oplus (E_l \otimes U)\\
	 &\quad \; (E_k^\dag \otimes U^\dag)\\
	 &= (E_k^\dag E_l \otimes U^\dag U) \oplus (U U^\dag \otimes E_l E_k^\dag)\\
	 &= (E_k^\dag E_l \otimes id_Y) \oplus (id_Y \otimes E_l E_k^\dag).
	\end{align*}
	By lemma \ref{lemma: gram matrix}, we can check  $((E_k^\dag E_l \otimes id_Y) \oplus (id_Y \otimes E_l E_k^\dag))_{k,l}$ is linearly independent or not by analysing its Gram matrix. Let $G$ be the Gram matrix of $(E_k^\dag E_l \oplus E_l E_k^\dag)_{k,l}$ and $G'$ be the Gram matrix of $((E_k^\dag E_l \otimes id_Y) \oplus (id_Y \otimes E_l E_k^\dag))_{k,l}$. The matrix elements of $G$ are
	\begin{align*}
		G_{(k,l),(k',l')} &= \langle E_k^\dag E_l \oplus E_l E_k^\dag, E_{k'}^\dag E_{l'} \oplus E_{l'} E_{k'}^\dag \rangle\\
		&= \langle E_k^\dag E_l, E_{k'}^\dag E_{l'} \rangle + \langle E_l E_k^\dag, E_{l'} E_{k'}^\dag \rangle.
	\end{align*}
	The matrix elements of $G'$ are
	\begin{align*}
		G'_{(k,l),(k',l')} &= \langle (E_k^\dag E_l \otimes id_Y) \oplus (id_Y \otimes E_l E_k^\dag), (E_{k'}^\dag E_{l'} \otimes id_Y) \oplus (id_Y \otimes E_{l'} E_{k'}^\dag) \rangle\\
		&= \langle E_k^\dag E_l \otimes id_Y, E_{k'}^\dag E_{l'} \otimes id_Y \rangle + \langle id_Y \otimes E_l E_k^\dag, id_Y \otimes E_{l'} E_{k'}^\dag \rangle\\
		&= \langle E_k^\dag E_l, E_{k'}^\dag E_{l'} \rangle \langle id_Y, id_Y\rangle + \langle id_Y, id_Y\rangle \langle E_l E_k^\dag, E_{l'} E_{k'}^\dag \rangle \\
		&= \dim Y (\langle E_k^\dag E_l, E_{k'}^\dag E_{l'} \rangle + \langle E_l E_k^\dag, E_{l'} E_{k'}^\dag \rangle)\\
		& = \dim Y \; G_{(k,l),(k',l')}.
	\end{align*}
	That is, $G' = \dim Y \; G$. Therefore $G$ is invertible if, and only if, $G'$ is invertible. By lemma \ref{lemma: gram matrix}, this says that $(E_k^\dag E_l \oplus E_l E_k^\dag)_{k,l}$ is linearly independent if, and only if, $((E_k^\dag E_l \otimes id_Y) \oplus (id_Y \otimes E_l E_k^\dag))_{k,l}$ is linearly independent. Then, by theorem \ref{theorem: characterization of extreme points for ucpt}, $\varepsilon$ is an extreme point of $UCPT(X)$ if, and only if, $\varepsilon \otimes \mathcal{U}$ is an extreme point of $UCPT(X \otimes Y)$.
\end{proof}

By theorem 4.23 of \cite{watrous_2018}, for $\dim X = 2$, the extreme points of $UCPT(X)$ are the unitary maps. Combining this result with proposition \ref{prop: tensor with unitary} we obtain the next theorem.

\begin{theorem}
	Let $\varepsilon$ be an extreme point of $UCPT(X)$ and $\varepsilon'$ be an extreme point of $UCPT(Y)$. If $\dim X = 2$ or $\dim Y = 2$, then $\varepsilon \otimes \varepsilon'$ is an extreme point of $UCPT(X \otimes Y)$.
\end{theorem}

For higher dimensions the extremality may not be preserved. One reason for this can be seen if we analyze theorem \ref{theorem: characterization of extreme points for ucpt}. Given $\varepsilon \in UCPT(X)$ with $CR(\varepsilon)$ linearly independent operation elements $(E_k)_k$, it is an extreme point of the UCPT maps if $(E_k^\dag E_l \oplus E_l E_k^\dag)_{k,l}$ is linearly independent as vectors of $\End(X)\oplus \End(X)$. The sequence $(E_k^\dag E_l \oplus E_l E_k^\dag)_{k,l}$ has $CR(\varepsilon)^2$ elements. For it to be linearly independent, $CR(\varepsilon)^2$ must be at most the dimension of $\End(X) \oplus \End(X)$, which is $2 (\dim X)^2$. Therefore, if $\varepsilon$ is extreme then $CR(\varepsilon) \leq \sqrt{2} \dim X$. Actually, a slightly better upper bound is given in \cite{ohno}, which says that $CR(\varepsilon) \leq \sqrt{2 (\dim X)^2-1}$. If we take two extreme UCPT maps of high enough Choi ranks, their tensor product can have a Choi rank that exceeds the previous upper bound, so the tensor product can't be extreme. This is proven in the next theorem.

\begin{theorem}
	Let $\varepsilon$ be an extreme point of $UCPT(X)$ and $\varepsilon'$ be an extreme point of $UCPT(Y)$. In particular, their Choi ranks satisfy $CR(\varepsilon) \leq \sqrt{2} \dim X$ and $CR(\varepsilon') \leq \sqrt{2} \dim Y$. If they further satisfy that $\sqrt[4]{2} \dim X < CR(\varepsilon)$ and $\sqrt[4]{2} \dim Y < CR(\varepsilon')$, then $\sqrt{2} \dim (X\otimes Y) < CR(\varepsilon \otimes \varepsilon')$, so $\varepsilon \otimes \varepsilon'$ isn't an extreme point of $UCPT(X \otimes Y)$.
\label{theorem: tensor product of ucpt need not be extreme}
\end{theorem}
\begin{proof}
	By corollary \ref{corollary: CR of tensor product}, we know that $CR(\varepsilon \otimes \varepsilon') = CR(\varepsilon)$ $CR(\varepsilon')$. If $\sqrt[4]{2} \dim X $ $<$ $ CR(\varepsilon)$ and $\sqrt[4]{2} \dim Y < CR(\varepsilon')$, multiplying both inequalities we get \linebreak $\sqrt[4]{2}\dim X \sqrt[4]{2} \dim Y < CR(\varepsilon)CR(\varepsilon')$, so $\sqrt{2} \dim (X \otimes Y) < CR(\varepsilon \otimes \varepsilon')$. As seen previously, this implies that $\varepsilon \otimes \varepsilon'$ isn't an extreme point of $UCPT(X \otimes Y)$.
\end{proof}

By theorem \ref{theorem: tensor product of ucpt need not be extreme}, to prove that the tensor product of extreme UCPT maps may not be extreme, it is sufficient to find extreme maps of high enough Choi rank. According to theorems 2.5 and 2.6 of \cite{ohno}, there is an extreme map $ \varepsilon_3 \in UCPT(\complexNumbers^3)$ with $CR(\varepsilon_3) = 4$ and an extreme map of $\varepsilon_4 \in UCPT(\complexNumbers^4)$ with $CR(\varepsilon_4) = 5$. More precisely, the operation elements (in the standard basis) for $\varepsilon_3$ are
\begin{equation}
	E_0 = 	\frac{1}{2}	\begin{pmatrix}
							1 & 0 & 0 \\
							0 & 0 & 0 \\
							0 & 0 & 0
						\end{pmatrix},
\end{equation}
\begin{equation}
E_1 = 	\frac{1}{2}	\begin{pmatrix}
	0 & 0 & 0 \\
	1 & 0 & 0 \\
	0 & \sqrt{2} & 0
\end{pmatrix},
\end{equation}
\begin{equation}
E_2 = 	\frac{1}{2}	\begin{pmatrix}
	0 & \sqrt{2} & 0 \\
	0 & 0 & \sqrt{3} \\
	0 & 0 & 0
\end{pmatrix},
\end{equation}
\begin{equation}
E_3 = 	\frac{1}{2}	\begin{pmatrix}
	0 & 0 & 1 \\
	0 & 0 & 0 \\
	\sqrt{2} & 0 & 0
\end{pmatrix}.
\end{equation}
The operation elements for $\varepsilon_4$ are
\begin{equation}
	F_0 = 	\frac{1}{2}	\begin{pmatrix}
		0 & 0 & 0 & 0\\
		0 & 0 & 1 & 0\\
		1 & 0 & 0 & 0\\
		0 & 0 & 0 & 0
	\end{pmatrix},
\end{equation}
\begin{equation}
F_1 = 	\frac{1}{2}	\begin{pmatrix}
	0 & 0 & 0 & 0\\
	0 & 0 & 0 & 0\\
	0 & 0 & 0 & \sqrt{2}\\
	0 & \sqrt{2} & 0 & 0
\end{pmatrix},
\end{equation}
\begin{equation}
F_2 = 	\frac{1}{2}	\begin{pmatrix}
	0 & 0 & \sqrt{3} & 0\\
	0 & 0 & 0 & 0\\
	0 & 0 & 0 & 0\\
	\sqrt{2} & 0 & 0 & 0
\end{pmatrix},
\end{equation}
\begin{equation}
F_3 = 	\frac{1}{2}	\begin{pmatrix}
	0 & 1 & 0 & 0\\
	0 & 0 & 0 & \sqrt{2}\\
	0 & 0 & 0 & 0\\
	0 & 0 & 0 & 0
\end{pmatrix},
\end{equation}
\begin{equation}
F_4 = 	\frac{1}{2}	\begin{pmatrix}
	0 & 0 & 0 & 0\\
	1 & 0 & 0 & 0\\
	0 & 1 & 0 & 0\\
	0 & 0 & 0 & 0
\end{pmatrix}.
\end{equation}
Notice that the above matrices are the complex conjugate of the ones defined in \cite{ohno}. This occurs because in \cite{ohno} a UCPT map is written as $\varepsilon(A) = \sum_k E_k^\dag A E_k$, while in this article the convention is to write $\varepsilon(A) = \sum_k E_k A E_k^\dag$, which is the same convention of \cite{nielsen_chuang_2010,watrous_2018}.

Both $\varepsilon_3$ and $\varepsilon_4$ have maximal Choi ranks. For both $n=3$ and $n=4$ we have $CR(\varepsilon_n) > \sqrt[4]{2} \, n$. Then, by theorem \ref{theorem: tensor product of ucpt need not be extreme}, we know that $\varepsilon_n \otimes \varepsilon_m$ is not an extreme UCPT map, for $n,m \in \{3,4\}$. It remains to show if we can apply theorem \ref{theorem: tensor product of ucpt need not be extreme} to obtain counterexamples for even higher dimensions.

\bibliographystyle{alpha}
\bibliography{on_the_extremality_of_the_tensor_product_of_quantum_channels}

\end{document}